\documentclass[12pt,reqno,a4paper]{amsart}    

\usepackage[totalwidth=500pt, totalheight=719pt]{geometry}

\usepackage{amsmath, amssymb}
\usepackage{amsthm}

\usepackage{graphicx}                      
\newcommand{\color}[2][{}]{}        

\usepackage{bbm}         

\theoremstyle{plain}            
\newtheorem{theorem}{Theorem}

\newtheorem{lemma}[theorem]{Lemma}

\theoremstyle{definition}       
\newtheorem{definition}[theorem]{Definition}

\theoremstyle{remark}
\newtheorem{remark}[theorem]{Remark}

\newcommand{\Sec}[1]{Section~\ref{sec:#1}}

\newcommand{\Thm}[1]{Theorem~\ref{thm:#1}}

\newcommand{\Lem}[1]{Lemma~\ref{lem:#1}}

\newcommand{\Lems}[2]{Lemmata~\ref{lem:#1}--\ref{lem:#2}}

\newcommand{\Rem}[1]{Remark~\ref{rem:#1}}

\newcommand{\Def}[1]{Definition~\ref{def:#1}}

\newcommand{\Fig}[1]{Figure~\ref{fig:#1}}

\numberwithin{equation}{section}

\DeclareMathOperator{\dist}   {dist}
\DeclareMathOperator{\dom}    {dom}

\DeclareMathOperator{\vol}    {vol}

\renewcommand{\Re}     {\mathrm {Re}\,}
\renewcommand{\Im}     {\mathrm {Im}\,}

\newcommand{\spec}[2][{}]   {\sigma_{\mathrm{#1}}(#2)}
\newcommand{\essspec}[1]{\spec[ess] {#1}}

\newlength{\maxbreite}%
\newlength{\maxhoehe}%
\newlength{\maxtiefe}%

\newcommand{\stelldrueber}[3][0pt]{
  \settowidth{\maxbreite}{#3}%
  \settoheight{\maxhoehe}{#3}%
  \settodepth{\maxtiefe}{#2}%
  \addtolength{\maxhoehe}{\maxtiefe}%
  {\makebox[\maxbreite]{\raisebox{\maxhoehe}{\hspace{#1}#2}}%
  \makebox[0pt][r]{#3}}%
}

\newcommand{\overcirc}[1]       
{\stelldrueber[.45ex]{$\scriptscriptstyle\circ$}{${#1}$}}

\newcommand{\R}{\mathbb{R}} 
\newcommand{\C}{\mathbb{C}} 
\newcommand{\N}{\mathbb{N}} 
\newcommand{\Sphere}{\mathbb{S}} 

\newcommand{\eps}{\varepsilon} 
\renewcommand{\phi}{\varphi}   
\newcommand{\e}{\mathrm e}  
\DeclareMathOperator{\dd}    {d\!}

\newcommand{\wt}{\widetilde}           
\newcommand {\qf}[1]{\mathfrak{#1}}    

\newcommand{\HS}{\mathcal H}           

\newcommand{\Sobsymb} {\mathsf H}      
\newcommand{\Contsymb} {\mathsf C}     
\newcommand{\Lsymb}    {\mathsf L}     

\newcommand{\Cont}[2][{}]{\Contsymb^{#1}({#2})}

\newcommand{\Lsqr}[2][{}]{\Lsymb_2^{#1}({#2})} 

\newcommand{\Sob}[2][1]{\Sobsymb^{#1}({#2})} 

\newcommand{\abs}[1]{\lvert #1 \rvert}    
\newcommand{\norm}[2][{}]{\|{#2}\|_{{#1}}}    
\newcommand{\normsqr}[2][{}]{\|{#2}\|^2_{#1}} 

\newcommand{\iprod}[3][{}]{\langle{#2},{#3}\rangle_{#1}}  

\newcommand{\set}[2]{\{ \, #1 \, | \, #2 \, \} } 
\newcommand{\bigset}[2]{\bigl\{ \, #1 \, \bigl|\bigr. \, #2 \, \bigr\} }
\newcommand{\Bigset}[2]{\Bigl\{ \, #1 \, \Bigl|\Bigr. \, #2 \, \Bigr\} }

\newcommand{\map}[3]{ #1 \colon #2 \longrightarrow #3 } 

\newcommand{\bd}  {\partial}                
\newcommand{\clo}[1]{\overline{{#1}}} 

\newcommand{\dcup}{\mathrel{\uplus}}               

\newcommand{\disjcup}{\mathrel{\overline{\dcup}}} 
\newcommand{\bigdisjcup}{\operatorname*{\overline{\biguplus}}}

\newcommand{\conj}[1]{\overline {{#1}}}       

\newcommand{\1}{\mathbbm 1}                    
\newcommand{\und}{\qquad\text{and}\qquad}
\newcommand{\Neu}{{\mathrm N}}              
\newcommand{\laplacian}[2][{}]{\Delta_{{#2}}^{{#1}}}

\newcommand{\EW}[3][{}]{\lambda^{{#1}}_{#2}({#3})}
\newcommand{\EWN}[2]{\EW[\Neu]{#1}{#2}}      

\newcommand{\ext}{{\mathrm{ext}}}
\newcommand{\inl}{{\mathrm{int}}}

\newcommand{\vxeps}{{\eps,v}}
\newcommand{\edeps}{{\eps,e}}

 \newcommand{\Err}{\mathcal O}

\DeclareMathOperator{\id}{id}
\newcommand{\Edge}{\Gamma^1}

\newcommand{\de}   {\mathrm d}         

\begin{document}

\title[Quantum networks modelled by graphs]{Quantum networks modelled
  by graphs}


\author{Pavel Exner}
\address{Department of Theoretical Physics, NPI,
  Academy of Sciences, 25068 \v{R}e\v{z} near Prague, and Doppler
  Institute, Czech Technical University, B\v{r}ehov\'{a}~7, 11519
  Prague, Czechia} 
\email{exner@ujf.cas.cz}

\author{Olaf Post}
\address{Institut f\"ur Mathematik,
  Humboldt-Universit\"at,
         Rudower Chaussee~25,
         12489 Berlin, Germany} 
\email{post@math.hu-berlin.de}




\begin{abstract}
  Quantum networks are often modelled using Schr\"odinger operators on
  metric graphs. To give meaning to such models one has to know how to
  interpret the boundary conditions which match the wave functions at
  the graph vertices. In this article we give a survey, technically not
  too heavy, of several recent results which serve this purpose. Specifically,
  we consider approximations by means of ``fat graphs'' --- in other words,
  suitable families of shrinking manifolds --- and discuss convergence
  of the spectra and resonances in such a setting.
\end{abstract}

\maketitle

%
\section{Introduction}
\label{sec:intro}
%

Quantum mechanics on metric graphs is a subject with a long history
which can be traced back to the paper of Ruedenberg and
Scherr~\cite{ruedenberg-scherr:53} on spectra of aromatic carbohydrate
molecules elaborating an idea of L.~Pauling. A new impetus came in the
eighties from the need to describe semiconductor graph-type
structures, cf.~\cite{exner-seba:89}, and the interest to these
problems driven both by mathematical curiosity and practical
applications is steadily growing; we refer
to~\cite{kostrykin-schrader:99, kuchment:04} or the
proceedings~\cite{bcfk:06} for a bibliography to the subject.

Since quantum graphs are supposed to model various real graph-like
structures with the transverse size which is small but non-zero, one
has to ask naturally how close are such system to an ``ideal'' graph
in the limit of zero thickness. This problem is not easy and a
reasonable complete answer is known in case of ``fat graphs'' with
Neumann boundary conditions and similar systems. A pioneering work in
this area was done by Freidlin and
Wentzell~\cite{freidlin-wentzell:93} and the
papers~\cite{kuchment-zeng:01} and~\cite{rubinstein-schatzman:01} can
be mentioned as important milestones. We managed to contribute to this
problem in a series of papers,~\cite{exner-post:05},~\cite{post:06}
and~\cite{exner-post:07}, in which we improved the approximation using
the intrinsic geometry of the manifold only, demonstrating the norm
resolvent convergence, and finally extending the approximation also to
resonances by means of complex scaling.

While these results provide in our opinion a solid insight into
the Neumann-type situation, we must acknowledge as the authors
that the three papers are long and rather technical, and some may
find them not easy to read. This motivated us to write the present
survey in which we intend to describe this family of approximation
results without switching in the heavy machinery; let the reader
judge whether we have succeeded.

Before proceeding let us mention that there is an encouraging recent
progress in the more difficult Dirichlet case, see \cite{acf:07,
  grieser:pre07, molchanov-vainberg:pre06b, cacciapuoti-exner:07},
however, we will not discuss it here.

Let us briefly describe the contents of the paper. In the next section
we describe the two basic objects of this paper, quantum graphs and
graph-like manifolds (cf.~\Fig{graph-mfd}). \Sec{disc-spec} is devoted
to convergence of the discrete spectrum summarizing the main results
of Ref.~\cite{exner-post:05}. An extension to non-compact graphs and a
resolvent convergence coming from~\cite{post:06} is given in
\Sec{res-conv}. Finally, in \Sec{reson-conv} describe the results of
Ref.~\cite{exner-post:07} showing how the resonances on quantum graphs
and graph-like manifolds approximate each other.


\begin{figure}
\label{fig:graph-mfd}
  \begin{picture}(0,0)%
    \includegraphics{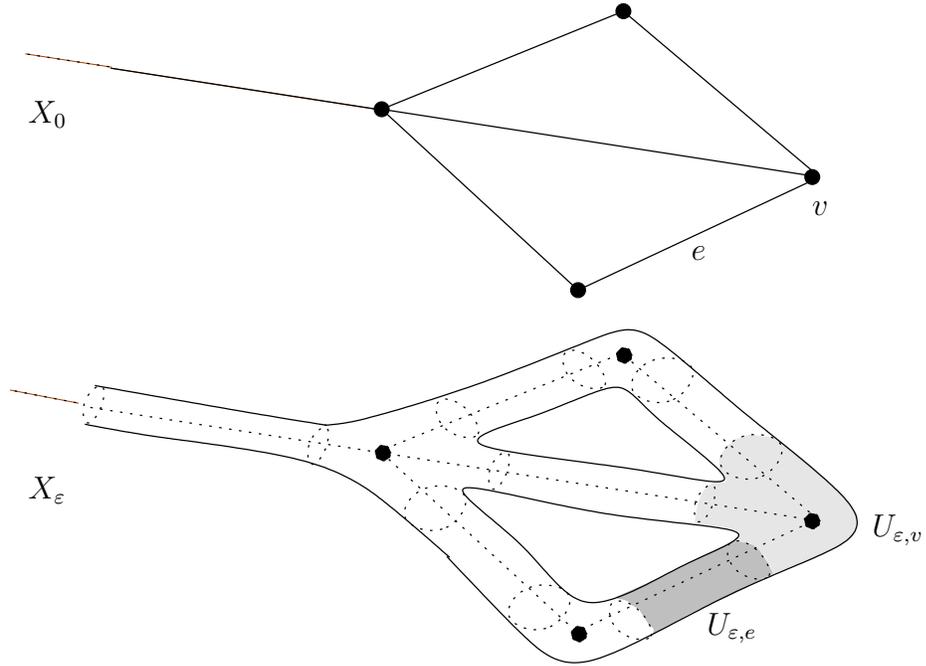}%
  \end{picture}%
  \setlength{\unitlength}{4144sp}%
  \begin{picture}(5172,3964)(532,-3415) \put(5344,-720) {$v$}%
    \put(4621,-991) {$e$}%
    \put(655,-171) {$X_0$}%
    \put(5704,-2645){$U_\vxeps$}%
    \put(4714,-3231){$U_\edeps$}%
    \put(655,-2419) {$X_\eps$}%
  \end{picture}
  \caption{The metric graph $X_0$ with one external edge, five
    internal edges and four vertices and the associated graph-like
    manifold $X_\eps$, here with cross section manifold
    $F=\Sphere^1$.}
\end{figure}

%
\section{Quantum networks and graphs}
\label{sec:qgraphs}
%

\subsection{Quantum graphs}
\label{sec:qg}

Suppose that $X_0$ is a connected metric graph given by
$(V,E,\bd,\ell)$ where $(V,E,\bd)$ is a usual graph, i.e., $V$ denotes
the set of vertices, $E$ denotes the set of edges, $\map \bd E {V
  \times V}$ associates to each edge $e$ the pair $(\bd_-e,\bd_+e)$ of
its initial and terminal point (and therefore an orientation). The
space $X_0$ being a \emph{metric} graph means that there is a
\emph{length function} $\map \ell E {(0,\infty]}$ associating to each
edge $e$ a length $\ell_e$. We often identify the edge $e$ with the
interval $(0,\ell_e)$ and use $x=x_e$ as a coordinate. In addition, we
denote $\dd x=\dd x_e$ the Lebesgue measure on $e$. In this way, $X_0$
becomes a topologically $1$-dimensional space with singularities at
the vertices.  Moreover, $X_0$ carries a natural metric by defining
the distance of two points to be the length of the shortest path in
$\Gamma$ joining these points.

We call an edge $e$ \emph{external} iff $\ell_e=\infty$ and
\emph{internal} otherwise and we denote the corresponding sets by
$E_\ext$ and $E_\inl$. Properly speaking, the map $\bd$ is only
defined on $E_\inl \times E_\inl$. For external edges, we do not
consider the ``end point'' at infinity as a vertex of $\Gamma$,
i.e., $\bd e$ contains only the initial vertex $\bd_-e=\bd e$. For
each vertex $v \in V$ we set
\begin{equation*}
  E_v^\pm := \set {e \in E} {\bd_\pm e = v} \qquad \text{and} \qquad
  E_v := E_v^+ \dcup E_v^-,
\end{equation*}
i.e., $E_v^\pm$ consists of all edges starting ($-$) resp.\ ending
($+$) at $v$ and $E_v$ is their disjoint union. Note that the
\emph{disjoint} union is necessary in order to allow single-edge
loops, i.e., edges having the same initial and terminal point. We
adopt the following uniform bounds on the degree $\deg v := |E_v|$
and the length function $\ell$:
\begin{subequations}
  \label{eq:deg.len}
  \begin{align}
    \label{eq:deg.bd}
    1 \le \deg v \le d_0&, \qquad v \in V,\\
    \label{eq:len.bd}
    \ell_e \ge \ell_0&, \qquad e \in E,
  \end{align}
\end{subequations}
where $1 \le d_0 < \infty$ and $0<\ell_0 \le 1$. Needless to say,
both the assumptions are fulfilled if $|E|$ and $|V|$ are finite.

We consider the Hilbert space defined naturally as the orthogonal
sum,
\begin{equation*}
  \HS_0 := \Lsqr {X_0} = \bigoplus_{e \in E} \Lsqr e.
\end{equation*}
A metric graph $X_0$ becomes a \emph{quantum} graph when it is
equipped with a self-adjoint operator $\laplacian {X_0}$ in the
Hilbert space $\Lsqr {X_0}$. The latter is assumed to act as
\begin{equation*}
  (\laplacian {X_0} f)_e := -f_e''
\end{equation*}
on each edge for $f \in \Sob[2] \Edge$, so a quantum graph is
determined by fixing the \emph{vertex boundary conditions} at each
vertex, in order to turn the formal Laplacian into a self-adjoint
operator. We will use the \emph{free} (often called, not quite
properly, \emph{Kirchhoff}) boundary conditions: a function lies
in the operator domain
\begin{equation*}
  \HS^2_0 := \dom {\laplacian {X_0}}
  \qquad \text{iff} \qquad
  f \in \bigoplus_{e \in E} \Sob[2] e
\end{equation*}
and the relations
\begin{subequations}
  \label{eq:kirchhoff}
  \begin{gather}
    \label{eq:kirchhoff.cont}
    f_{e_1} (v) = f_{e_2}(v), \qquad e_1, e_2 \in E_v\\
    \label{eq:kirchhoff.der}
    \sum_{e \in E_v} \vec f_e'(v) = 0
  \end{gather}
\end{subequations}
are fulfilled for all $v \in V$ where
\begin{equation}
  \label{eq:der.orient}
  \vec f_e'(v) :=
  \begin{cases}
    -f_e'(0), & \text{if $v=\bd_-e$,}\\
    +f_e'(\ell_e), & \text{if $v=\bd_+e$}
  \end{cases}
\end{equation}
defines the \emph{inward} derivative of $f'_e$ at $v$. Under the
assumptions~\eqref{eq:deg.len}, the operator $\laplacian{X_0}$ is
self-adjoint in $\HS_0$ (cf.~\cite{kuchment:04}).  The quadratic form
associated with the operator $\laplacian {X_0}$ is $\qf d_0(f):=
\normsqr{(\laplacian {X_0})^{1/2} f}$ and can be expressed as
\begin{equation*}
  \qf d_0 (f) = \normsqr{f'} = \sum_{e \in E} \normsqr{f'_e}
  \qquad \text{on} \qquad
  \HS^1_0 := \Sob{X_0} := \Cont {X_0} \cap \bigoplus_{e \in E} \Sob e.
\end{equation*}

\subsection{Graph-like manifolds}
\label{sec:mfd}
Let us pass to a model of a quantum network which we will
consider, corresponding to the idea that the graph has a small,
but non-zero thickness.  Let $X_\eps$ be a $d$-dimensional
connected manifold with metric $g_\eps\:$ (cf.~\Fig{graph-mfd}).
If $X_\eps$ has boundary, we denote it by $\bd X_\eps$; let us
stress that our discussion covers different kind of models,
``full'' fat graphs considered by Kuchment and Zeng
\cite{kuchment-zeng:01, kuchment-zeng:03} where the boundary is
present, as well as ``hollow'' or ``sleeve-type'' manifolds having
no boundary. We assume that $X_\eps$ can be decomposed into open
sets $U_\edeps$ and $U_\vxeps$, i.e.,\footnote{
  \label{fn:dcup} Here and in the following, the expression
  $A=\bigdisjcup_i A_i$ means that the $A_i$'s are open (in $A$),
  mutually disjoint and the interior of $\bigcup_i \clo A_i$ equals
  $A$; recall that in an $\Lsymb_2$-theory it is enough to have charts
  covering a set of full measure. }
\begin{equation}
  \label{eq:decomp}
  X_\eps = \bigdisjcup_{e \in E} U_\edeps \,\disjcup\,
  \bigdisjcup_{v \in V} U_\vxeps.
\end{equation}
Denote the metric on $X_\eps$ by $g_\eps$.  We assume that
$U_\edeps$ and $U_\vxeps$ are isometric to
\begin{subequations}
  \label{eq:met}
  \begin{align}
    U_\edeps & \cong (e \times F, g_\edeps) &
    g_\edeps &= \de x_e^2 + \eps^2 h\\
    U_\vxeps & \cong (U_v, g_\vxeps) & g_\vxeps &= \eps^2 g_v
  \end{align}
\end{subequations}
where $(F,h)$ is a compact $m$-dimensional manifold with $m:=(d-1)$,
and $(U_v,g_v)$ is an $\eps$-independent $d$-dimensional
manifold\footnote{We employ just these particular charts; there is no
  need for a complete system. Strictly speaking, \Fig{graph-mfd} shows
  a slightly different situation where the edge neighborhoods are
  shortened (cf.\ \Rem{short.ed} below).}. Note that $\bd U_\vxeps
\setminus \bd X_\eps$ has $(\deg v)$-many components isometric to
$(F,\eps^2 h)$ denoted by $(\bd_e U_v, \eps^2 h)$ for $e \in E_v$.  We
sometimes write $U_e:= e \times F$.
\begin{remark}
  \label{rem:collar}
  For technical reasons (cf.~\Lem{cn}) we assume that near $\bd_e
  U_v$, the (unscaled) manifold $(U_v, g_v)$ has a collar neighborhood
  $((0,\ell_0/2) \times F, \de \hat x^2 + \eps^2 h)$.  If such a
  collar neighborhood happens to be too small, we can just change the
  decomposition~\eqref{eq:decomp} in such a way that we add a cylinder
  of length $\eps \ell_0/2$ (the length taken in the edge coordinates
  $x$ on $U_e$) from the edge neighborhood to the vertex neighborhood
  (becoming here a cylinder of length $\ell_0/2$ in the vertex
  coordinates on $U_v$ since $\hat x = x/\eps$).
\end{remark}

The cross section manifold $F$ has a boundary or does not have one,
depending on the analogous property of $X_\eps$. For simplicity, we
suppose that $\vol_m F=1$.  Clearly, we have
\begin{equation}
  \label{eq:met.prod.asym}
  \dd U_\edeps = \eps^m \dd F \dd x_e
\end{equation}
for the Riemannian densities.  We consider the Hilbert space
\begin{equation*}
  \HS_\eps = \Lsqr {X_\eps}
\end{equation*}
and the Laplacian $\laplacian {X_\eps} \ge 0$ (with Neumann
boundary conditions if $\bd X_\eps \ne \emptyset$) defined on
\begin{equation*}
  \HS^2_\eps := \dom \laplacian{X_\eps}.
\end{equation*}
The associated quadratic form $\qf d_\eps(u):=
\normsqr{(\laplacian {X_\eps})^{1/2} u}$ can be expressed as
\begin{equation*}
  \qf d_\eps (u) = \normsqr{\de u}
  \qquad \text{on} \qquad
  \HS^1_\eps := \Sob{X_\eps}
\end{equation*}
where $\dd u$ is the exterior derivative of $u$. Note that the
quadratic form expression contains only the metric $g_\eps$, but
no derivatives of $g_\eps$.

\begin{remark}
  \label{rem:short.ed}
  We have chosen the full edge length on the edge neighborhood
  although this assumption is not valid if $X_\eps$ is the
  $\eps/2$-neighborhood of a metric graph $X_0$ \emph{embedded} into a
  Euclidean space from which it inherited its metric. In such a
  situation, however, the metric $g_\edeps$ differs from the metric
  $\de x_e^2 + \eps^2 h$ only by a small longitudinal error.  Using
  the fact that the Laplacian on $X_\eps$ defined via its quadratic
  form depends only on $g_\eps$ (and not on its derivatives), it can
  be shown that a small (uniform) perturbation of the product
  structure has only a small effect on the Laplacian, its spectrum,
  etc.
\end{remark}
In addition, we assume that the following uniformity conditions
are valid,
\begin{align}
  \label{eq:vol.ew}
  c_{\vol} := \sup_{v \in V}\: {\vol_d U_v} < \infty, \qquad \qquad
  \lambda_2 := \inf_{v \in V} \EWN 2 {U_v} > 0,
\end{align}
where $\EWN 2 {U_v}$ denotes the second (i.e., first non-zero)
Neumann eigenvalue of $(U_v,g_v)$.

Roughly speaking, the requirements~\eqref{eq:vol.ew} mean that the
region $U_v$ remains small w.r.t. the vertex index --- see the
discussion in~\cite[Rem.~2.7]{post:06}) for more details. Needless
to say, these assumptions are trivially satisfied once the vertex
set $V$ is finite.

Recall also one more domain related to the Laplacian for $\eps \ge 0$.
We denote by
\begin{equation*}
  \HS_\eps^k := \dom (\laplacian{X_\eps}+1)^{k/2}
\end{equation*}
the scale of Hilbert spaces associated with the self-adjoint,
non-negative operator $\laplacian{X_\eps}$, together with its
natural norm
\begin{equation*}
  \norm[k] u := \norm{(\laplacian{X_\eps}+1)^{k/2}u}.
\end{equation*}
For negative exponents, we set
\begin{equation*}
  \HS_\eps^{-k} := \big(\HS_\eps^k \big)^*\,;
\end{equation*}
note that $\HS_\eps^{-k}$ can be again viewed as the completion of
$\HS_\eps \subset \HS_\eps^{-k}$ with respect to the appropriate norm,
in this case $\norm[-k] u = \norm{( \laplacian{X_\eps}+1)^{-k/2} u}$.

%
\section{Convergence of discrete spectrum}
\label{sec:disc-spec}
%
Let us start with a simple thing, a convergence result using a
variational characterization of discrete eigenvalues. We assume
that $X_0$ is a compact metric graph, in other words,
\begin{equation*}
  |E| < \infty \quad \text{and} \quad
  \ell_e < \infty \quad\text{for all $\:e \in E$.}
\end{equation*}
In this case, the graph-like manifold $X_\eps$ is also compact, and
the spectrum of $\laplacian {X_\eps}$, $\eps \ge 0$, is purely
discrete. We denote by $\lambda_k(\eps)$ the $k$-th eigenvalue of
$\laplacian {X_\eps}$, $k \ge 1$, repeated with respect to the
multiplicity. The min-max variational characterization is then
\begin{equation}
  \label{eq:max.min}
  \lambda_k(\eps)
  = \inf_{L_k} \sup_{u \in L_k \setminus \{0\} }
  \frac {\qf d_\eps(u)}{\normsqr u}
\end{equation}
where the infimum is taken over all $k$-dimensional subspaces $L_k$ of
$\Sob {X_\eps}$.

Our main result in this section, coming
from~\cite{rubinstein-schatzman:01,kuchment-zeng:01,exner-post:05},
is the following:
\begin{theorem}
  \label{thm:disc}
  Assume that the metric graph $X_0$ is compact, i.e., the graph is
  finite and each edge has a finite length. Then the $k$-th (Neumann)
  eigenvalue of the Laplacian on the graph-like manifold
  $\lambda_k(\eps)$ converges to the $k$-th eigenvalue of the
  Laplacian on the metric graph $X_0$ with the free boundary conditions.
  Specifically, we have
  \begin{equation*}
    \lambda_k(\eps) - \lambda_k(0) = \Err(\eps^{1/2})
  \end{equation*}
  where the error term $\Err(\eps^{1/2})$ depends on the index $k$.
\end{theorem}
The proof is based on an abstract comparison result.  In order to
compare the respective eigenvalues of the quantum graph and the
graph-like manifold, we need identification maps expressed in
terms of the quadratic forms, namely
\begin{equation}
  \label{eq:j.scale1}
  \map {J^1} {\HS_0^1} {\HS_\eps^1}
  \und
  \map {J^1{}'} {\HS_\eps^1}  {\HS_0^1} .
\end{equation}
We have the following abstract eigenvalue comparison result (see,
e.g.,~\cite[Lem.~2.1]{exner-post:05}):
\begin{lemma}
  Assume that there are $\delta_1=\delta_1(\eps)$ and $\delta_2=\delta_2(\eps)$
such that
  \begin{subequations}
    \label{eq:j1.est}
    \begin{align}
      \label{eq:j1.est1}
      \qf d_0(f) + \delta_1 \normsqr[1] f   &\ge \qf d_\eps (J^1 f), &
      \normsqr f - \delta_1 \normsqr[1] f   &\le \normsqr {J^1 f},\\
      \label{eq:j1.est2}
      \qf d_\eps (u) + \delta_2 \normsqr[1] u &\ge \qf d_0    (J^1{}'u),&
      \normsqr u     - \delta_2 \normsqr[1] u &\le \normsqr {J^1{}' u}
    \end{align}
  \end{subequations}
  for all $f \in \HS_0^1$ and $u \in \HS_\eps^1$. Then
  \begin{equation*}
      \frac{-2(1+\lambda_k(0))(1+\delta_1)}
       {1- \bigl(\delta_1+\delta_2(1+\delta_1) \bigr)(1+\lambda_k(0))}
       \, \delta_2
    \le \lambda_k(\eps) - \lambda_k(0)
    \le \frac{2 (1+\lambda_k(0))}
       {1- \delta_1(1+\lambda_k(0))}
       \, \delta_1
  \end{equation*}
  where the upper bound is $\Err(\delta_1)$ depending on
  $\lambda_k(0)$ and the lower bound is $\Err(\delta_2)$ depending
  additionally on $\delta_1$.
\end{lemma}
\medskip

\noindent In our concrete example, the ${\HS_0^1}\to {\HS_\eps^1}$
identification operator can be chosen as
\begin{equation}
  \label{eq:j1}
  J^1 f (z):=
  \begin{cases}
    \eps^{-m/2} f_e(x) & \text{if $z=(x,y) \in U_e$},\\
    \eps^{-m/2} f(v)   & \text{if $z \in U_v$}
  \end{cases}
\end{equation}
Note that the definition makes sense since functions in $\HS^1_0$
are continuous. For the map in the opposite direction, we first
introduce the following averaging operators
\begin{gather*}
  (N_e u)(x) :=
   \iprod[F]{\1} {u_e(x,\cdot)} =
   \int_F u_e(x,y) \dd F(y),\\
   C_v u :=
   \iprod[U_v] {\1_v} {u_v} =
   \frac 1 {\vol_d U_v} \int_{U_v} u \dd U_v
\end{gather*}
for $u \in \wt \HS = \Lsqr {X_\eps}$. What they yield is nothing else
than the Fourier coefficient corresponding to the first (transverse)
eigenfunction $\1(y)=1$ on $F$ and $\1_v(z)=(\vol_d U_v)^{-1}$ on
$U_v$, respectively; note that these eigenfunctions are constant and
that $\vol_m F = 1$. We set
\begin{equation}
  \label{eq:j1'}
  (J^1{}' u)_e (x):=
    \eps^{m/2} \Bigl(
        N_e u (x) + \sum_{v \in \bd e}
        \rho(\dist(x,v)) \bigl( C_v u - N_e u (v) \bigr) \Bigr)
\end{equation}
for $x \in e$, where $\map \rho \R {[0,1]}$ is a smooth function
with
\begin{equation}
  \label{eq:def.cut.off}
  \rho(0)=1           \quad \text{and} \quad
  \rho(r)=0
      \quad \text{for all $r \ge \ell_0/2$}.
\end{equation}
The interpolating contribution related to $C_v u - N_eu(v)$ is
needed in order to make $J^1{}'u$ continuous at each vertex. The
following lemma ensures that the error coming from this correction
remains small:
\begin{lemma}
  \label{lem:cn}
  We have
  \begin{equation*}
    \eps^m |C_v u - N_eu(v)|^2
    \le \frac 8 {\ell_0}\Bigl( 1 + \frac 1 {\lambda_2}\Bigr)\eps
            \normsqr[\vxeps] {\de u}
  \end{equation*}
  for all $u \in \HS_\eps^1$ and $v \in \bd e$, where
  $\norm[\vxeps]\cdot$ denotes the $\Lsymb_2$-norm on $U_\vxeps$.
\end{lemma}
\begin{proof}
  Due to \Rem{collar}, each component $\bd_e U_v$ of $\bd U_v$ has a
  collar neighborhood $(0,\ell_0/2)\times F$ of length $\ell_0/2$ (in
  the unscaled coordinates of $U_v$).  The Cauchy-Schwarz inequality
  and the Sobolev trace estimate (see e.g.~\cite[Lem.~8]{kuchment:04})
  \begin{equation}
    \label{eq:sob.est}
    \abs{f(0)}^2
         \le \frac 8 {\ell_0}
    \int_0^{\ell_0/2} \bigl( \abs{f(s)}^2 + \abs{f'(s)}^2 \bigr) \dd s
  \end{equation}
  yield (with $f(s)=f_y(s)=C_v u - u(s,y)$ in the collar coordinates,
  and then integration over $y \in F$)
  \begin{equation*}
    |C_v u - N_eu(v)|^2
    \le \int_F |C_v u - u(0,\cdot)|^2 \dd F
    \le \frac 8 {\ell_0}
    \bigl( \normsqr[v]{C_v u - u} + \normsqr[v]{\de u} \bigr),
  \end{equation*}
  where $C_v u$ is considered to be a constant function on $U_v$ and
  $\norm[v]\cdot$ is the $\Lsymb_2$-norm on $U_v$.  Now $C_v u - u$ is
  orthogonal to the first (constant) eigenfunction of the Neumann
  Laplacian on $U_v$, and the min-max principle ensures that the
  squared norm of $C_v u - u$ can be estimated by $(\EWN 2 {U_v})^{-1}
  \normsqr[U_v]{\de u}$.  Using the scaling $g_\vxeps=\eps^2 g_v$ of
  the metric and~\eqref{eq:vol.ew}, we obtain the desired estimate.
\end{proof}

We also have to make sure that eigenfunctions $u$ of
$\laplacian{X_\eps}$ belonging to eigenvalues \emph{bounded} with
respect to $\eps$, cannot concentrate on the vertex neighborhoods:
\begin{lemma}
  \label{lem:vx}
  We have
  \begin{equation*}
    \normsqr[\vxeps] u
    \le c_{\mathrm {vx}} \eps
       \bigl( \normsqr[\vxeps] {\de u} + \normsqr[\edeps] u
             + \normsqr[\edeps] {\de u}
       \bigr)
  \end{equation*}
  for $u \in \HS_\eps^1$, where $e$ is any edge adjacent to the vertex
  $v$. The constant $c_{\mathrm{vx}}$ in this inequality depends only
  on $\ell_0$, $\lambda_2$ and $c_{\vol}$.
\end{lemma}
\begin{proof}
  We employ the estimate
  \begin{equation*}
    \norm[\vxeps] u
    \le \norm[\vxeps] {u - C_vu} +
      \sqrt {\vol_d U_\vxeps}
      \bigl( |C_v u - N_e u(v)| + |N_e u(v)| \bigr).
  \end{equation*}
  The first summand can be treated as in the previous proof, the
  second by \Lem{cn}, and the last one by a Sobolev trace estimate on
  the \emph{edge} neighborhood similar to~\eqref{eq:sob.est}, namely
  $\eps^m|N_e u(v)|^2 \le 8 (\normsqr[\edeps] u + \normsqr[\edeps]{\de
    u})/\ell_0$.
\end{proof}

\begin{proof}[Proof of \Thm{disc}]
  It remains to show that the conditions~\eqref{eq:j1.est} are
  fulfilled. We do not give the details here referring
  to~\cite[Sec.~5]{exner-post:05}. The proof
  of~\eqref{eq:j1.est1} is simple, and it works even with
  $\delta=0$, hence one obtains a stronger estimate,
  $\lambda_k(\eps) \le \lambda_k(0)$.

  For the opposite inequality, we need to verify~\eqref{eq:j1.est2}.
  To this end we need \Lem{cn} in the norm and a quadratic form estimate.
  The norm estimate uses in addition \Lem{vx}, and the estimate
  \begin{equation*}
    \normsqr[\edeps] u - \normsqr[\edeps] {Nu}
    \le \Err(\eps) \normsqr[1] u
  \end{equation*}
  which follows from~\eqref{eq:trans}. For the quadratic form estimate
  we need the simple Cauchy-Schwarz bound $\normsqr[\edeps]{\de u}
  \le \normsqr[\edeps]{(Nu)'}$.
\end{proof}

\begin{remark}
  Similar results can be obtained for more general situations when the
  vertex and edge neighborhoods scale at different rates,
  cf.~\cite{exner-post:05, kuchment-zeng:03}, and in certain
  situations also for the Dirichlet Laplacian, cf.~\cite{post:05},
  where, however, the resulting graph operator is decoupled.
\end{remark}

%
\section{Resolvent convergence}
\label{sec:res-conv}
Next we would like to go further and prove also results for
non-compact graphs, and also convergence of eigenfunctions or
resolvents. To do so we need some more notation. We write
$\HS_\eps$ and $\Delta_\eps=\laplacian{X_\eps}$ for the
$\eps$-dependent spaces, $\eps>0$. We stress that the parameter
$\eps$ enters only through the quantity $\delta=\delta_\eps>0$ and
one can interpret it as a label for the second Hilbert space
involved -- see also the appendix in~\cite{post:06} for the
concept of a ``distance'' between two Hilbert spaces and
associated non-negative operators. For brevity, we set
\begin{equation*}
  R_\eps := (\Delta_\eps + 1)^{-1}, \qquad
  \Delta_\eps := \laplacian{X_\eps}, \qquad \eps \ge 0.
\end{equation*}
\begin{definition}
  \label{def:quasi}
  We say that an operator $\map J {\HS_0} {\HS_\eps}$ is
  \emph{$\delta$-quasi-unitary} with respect to $\Delta_\eps$, iff
  $J^*J=\id_0$, $\norm J = 1$, and
  \begin{equation*}
    \norm[1 \to 0] {JJ^*-\id_\eps}
    = \norm{(JJ^*-\id_\eps)R_\eps^{1/2})} \le \delta,
  \end{equation*}
  where $\id_\eps$ is the identity on $\HS_\eps$, and $\norm[1 \to
  0] A$ is the operator norm of $\map A {\HS_\eps^1} {\HS_\eps^0}$.
\end{definition}
In our particular situation, we will employ the quasi-unitary
operator
\begin{equation}
  \label{eq:j}
  J f := \{ f_e \otimes \1_\eps \}_e \oplus \{ 0_v\}_v,
\end{equation}
where $\1_\eps=\eps^{-m/2}\1$ is the lowest normalized
eigenfunction on $(F,\eps^2h)$ and in turn $0_v$ is the zero
function on $U_v$. The quasi-unitarity is stated in the following
lemma:
\begin{lemma}
  \label{lem:j.quasi}
  The map $J$ defined by (\ref{eq:j}) is $\delta_\eps$-quasi-unitary, where
  $\delta_\eps=\Err(\eps^{1/2})$ depends only on $\ell_0$, $c_{\vol}$
  and $\lambda_2$.
\end{lemma}
\begin{proof}
  A simple calculation shows that $J^*J=\id_0$, $\norm J = 1$, and that
  \begin{equation*}
    \normsqr{JJ^* u - u}
    =   \sum_{e \in E} \normsqr[\edeps]{N_e u - u}
      + \sum_{v \in V} \normsqr[\vxeps] u.
  \end{equation*}
  The function $N_e u(x) - u(x,\cdot)$ is orthogonal to the constant
  function on $F$, and by the min-max principle we infer that
  \begin{equation}
    \label{eq:trans}
    \normsqr[\edeps] {N_e u - u}
    \le \frac {\eps^2} {\EW 2 F} \, \normsqr[\edeps]{\de_F u}
    \le \frac {\eps^2} {\EW 2 F} \, \normsqr[\edeps]{\de u}
  \end{equation}
  where $\EW 2 F$ is the first non-zero (Neumann) eigenvalue on $F$,
  and $\de_F$ is the derivative with respect to the transverse
  variable(s).  The estimate of the sum over the vertex contributions
  follows from \Lem{vx}.
\end{proof}

We also need a tool to compare the Laplacians on $\HS_0$ and
$\HS_\eps$. To this end we put:
\begin{definition}
  We say that $\Delta_\eps$ and $\Delta_0$ are \emph{$\delta$-close}
  w.r.t. the map $\map J {\HS_0}{\HS_\eps}$ iff
  \begin{equation*}
    \norm[2 \to -2] {J\Delta_0 - \Delta_\eps J}
    = \norm{R_\eps J - J R_0}
     \le \delta,
  \end{equation*}
  where $\norm[2 \to -2] A$ denotes the operator norm of $\map A {\HS_0^2}
  {\HS_\eps^{-2}}$.
\end{definition}
\begin{remark}
  Note that a $0$-quasi-unitary map is indeed unitary. Furthermore, if
  $\Delta_0$ and $\Delta_\eps$ are $0$-close with respect to a
  $0$-quasi-unitary map $J$, then $\Delta_0$ and $\Delta_\eps$ are
  unitarily equivalent. In this sense, the concept of quasi-unitarity
  and closeness provides a quantitative way to measure how far a pair
  of operators is from being unitarily equivalent.
\end{remark}
In order to show that the operators $\Delta_\eps$ and $\Delta_0$
are $\delta$-close, it is often easier to deal with the respective
quadratic form domains as we have already done when demonstrating
the convergence of the discrete spectrum. We thus want to compare
the identification operators on the scale of order $1$ with the
quasi-unitary map $J$:
\begin{definition}
  We say that the identification maps~\eqref{eq:j.scale1} are
  \emph{$\delta$-compatible} with the map $\map J {\HS_0} {\HS_\eps}$
  iff
  \begin{equation*}
    \norm[1 \to 0]{J-J^1} = \norm{(J-J^1)R_0^{1/2}} \le \delta
       \quad \text{and} \quad
    \norm[1 \to 0]{J^*-J^1{}'}
    = \norm{(J^*-J^1{'})R_\eps^{1/2}} \le \delta.
  \end{equation*}
\end{definition}
By means of an adjoint we obtain from $J^1{}'$ a natural map
\begin{equation*}
  \map{J^{-1}:=(J^1{}')^*}{\HS_0^{-1}}{\HS_\eps^{-1}}.
\end{equation*}
Now it is easy to derive the following criterion for
$\delta$-closeness:
\begin{lemma}
  \label{lem:j.comm1}
  Assume that $J^1$ and $J^1{}'$ are $\delta$-compatible w.r.t.
  the map $J$, and that
  \begin{equation}
    \label{eq:j.comm1}
    \norm[1 \to -1] {J^{-1}\Delta_0 - \Delta_\eps J^1}
    = \norm{R_\eps^{1/2}( J^{-1}\Delta_0
               - \Delta_\eps J^1) R_0^{1/2}}
     \le \delta.
  \end{equation}
  Then $\Delta_\eps$ and $\Delta_0$ are $3\delta$-close with respect
  to $J$.
\end{lemma}

We first check that the identification maps $J^1$ and $J^1{}'$ are
indeed $\Err(\eps^{1/2})$-compatible with respect to the map $J$:
\begin{lemma}
  \label{lem:j.comp}
  The maps $J^1$ and $J^1{}'$ as defined in~\eqref{eq:j1}
  and~\eqref{eq:j1'} are $\delta_\eps$-compatible with $J$, where
  $\delta_\eps=\Err(\eps^{1/2})$ depends only on $\ell_0$, $d_0$,
  $c_{\vol}$ and $\lambda_2$.
\end{lemma}
\begin{proof}
  We have
  \begin{equation*}
    \normsqr{(J-J^1)f}
    = \eps^{-m} \sum_{v \in V} (\vol_d U_\vxeps) |f(v)|^2
    = \eps  \sum_{v \in V} (\vol_d U_v) |f(v)|^2 .
  \end{equation*}
  By a standard Sobolev estimate -- see, e.g.,
  estimate~\eqref{eq:sob.est} or~\cite[Lem.~2.4]{post:06}) -- we can
  estimate the latter sum by $8 \eps c_{\vol} \normsqr[1] f/\ell_0$.
  For the other identification operator we have
  \begin{equation*}
    \normsqr{(J^*-J^1{}')u}
    =  \sum_{e \in E} \sum_{v \in \bd e}
         \int_0^{\ell_0/2} \rho(r)^2 \,
          \eps^m \bigl| C_v u - N_e u(v) \bigr|^2\, \dd r.
  \end{equation*}
  Using now \Lem{cn} and reordering the sum, we obtain the additional
  factor $d_0$, the maximum degree of a vertex, and the second
  estimate follows as well.
\end{proof}

Next, we will now indicate briefly how to prove the closeness of
the Laplacians:
\begin{lemma}
  \label{lem:j.comm}
  The Laplacians $\Delta_\eps$ and $\Delta_0$ are $\delta_\eps$-close
  with respect to the map $J$ defined in~\eqref{eq:j} where
  $\delta_\eps=\Err(\eps^{1/2})$ depends only on $\ell_0$, $d_0$ and
  $\lambda_2$.
\end{lemma}
\begin{proof}
  We check the condition~\eqref{eq:j.comm1} of \Lem{j.comm1} which
  reduces to estimating
  \begin{equation*}
    \bigl|\qf d_0(f, J^1{}' u) - \qf d_\eps( J^1 f, u) \bigr|
    \le \sum_{e \in E} \sum_{v \in \bd e}
      \norm[e]{f'} \int_0^{\ell_0/2} |\rho'(r)|^2 \dd r \,
        \eps^{m/2} \bigl| C_v u - N_e u(v) \bigr|.
  \end{equation*}
  in terms of $\delta_\eps \norm[1]f \norm[1]u$; the claim follows from
  \Lem{cn} and Cauchy-Schwarz.
\end{proof}
Putting together the previous results, we come to the following
conclusion:
\begin{theorem}
  \label{thm:res}
  Adopt the uniformity conditions~\eqref{eq:deg.len}
  and~\eqref{eq:vol.ew}. Then the Laplacians $\laplacian {X_\eps}$ and
  $\laplacian {X_0}$ are $\Err(\eps^{1/2})$-close with respect to the
  quasi-unitary map $J$ defined in~\eqref{eq:j}, i.e.
  \begin{equation*}
    \norm{(\laplacian {X_\eps}+1)^{-1} J
          - J (\laplacian {X_0} + 1)^{-1}}
    \le \Err(\eps^{1/2}),
  \end{equation*}
  where the error term depends only on $\ell_0$, $d_0$, $c_{\vol}$ and
  $\lambda_2$. In addition, we have
  \begin{equation}
    \label{eq:sandwich}
    \norm{(\laplacian {X_\eps}+1)^{-1}
          - J (\laplacian {X_0} + 1)^{-1} J^*}
    \le \Err(\eps^{1/2}).
  \end{equation}
\end{theorem}
\begin{proof}
  The first estimate follows from \Lems{j.comm1}{j.comm}; the second
  one in turn is a consequence of the first estimate and \Lem{j.quasi}.
\end{proof}
One can now develop the standard functional calculus for the
Laplacians $\Delta_\eps$ and $\Delta_0$, and deduce estimates
similar to the ones in \Thm{res}, but with the resolvent replaced
by more general functions $\phi(\Delta_\eps)$ of the Laplacians.
Specifically, $\phi$ need to be measurable, continuous in a
neighborhood of the spectrum of $\Delta_0$, and the limit at
infinity must exist. For example, one can control the heat
operators via $\phi_t(\lambda)=\e^{-t\lambda}$ or the spectral
projectors via $\phi=\1_I$. A proof of the following result can be
found in the appendices of the paper~\cite{post:06}, see
also~\cite{exner-post:07}:
\begin{theorem}
  \label{thm:res.spec}
  Under the assumptions of the previous theorem, we have
  \begin{equation*}
    \norm{ \1_I(\laplacian {X_\eps}) J
          - J \1_I(\laplacian {X_0})}
    \le \Err(\eps^{1/2})
      \quad \text{and} \quad
    \norm{ \1_I(\laplacian {X_\eps})
          - J \1_I(\laplacian {X_0})J^*}
    \le \Err(\eps^{1/2})
  \end{equation*}
  for the spectral projections provided $I$ is a compact interval such
  that $\bd I \cap \spec {\laplacian {X_0}} = \emptyset$. In
  particular, if $I$ contains a single eigenvalue $\lambda(0)$ of
  $\laplacian {X_0}$ with multiplicity one corresponding to an
  eigenfunction $u(0)$, then
  there is an eigenvalue $\lambda(\eps)$ and an eigenfunction
  $u(\eps)$ of $\laplacian{X_\eps}$ such that
  \begin{equation*}
    \norm{J u(0) - u(\eps)} = \Err(\eps^{1/2}).
  \end{equation*}
  In addition, the spectra converge uniformly on $[0,\Lambda]$, i.e.
  \begin{equation*}
    \spec {\laplacian {X_\eps}} \cap [0,\Lambda] \to
    \spec {\laplacian {X_0}} \cap [0,\Lambda]
  \end{equation*}
  in the sense of Hausdorff distance on compact subsets of $[0,\Lambda]$. The
  same result is true if we consider only the essential or the
  discrete spectral components.
\end{theorem}
Naturally, the above stated spectral convergence reduces to the
claim of \Thm{disc} in the situation when the spectra are purely
discrete.

%
\section{Convergence of resonances}
\label{sec:reson-conv}
%
In the final section we will deal with the convergence of resonances
in the present setting. It is useful to include into the
considerations also eigenvalues embedded in the continuous spectrum,
because it may happen that resonances of a ``fat graph'' converge to
such an eigenvalue, as it can be seen, e.g., in a simple motivating
example of the metric graph consisting of a single loop with a
half-line ``lead'' attached \cite{exner-post:07}.

A standard and successful method of dealing with resonances is
based on the concept of \emph{complex scaling}, often an
\emph{exterior} one. The method has its roots in the seminal
papers~\cite{aguilar-combes:71, balslev-combes:71} and a lot of
work was devoted to it; we refer to~\cite{exner-post:07} for a
sample bibliography. The main virtue is that it allows to
reformulate treatment of resonances, i.e. poles of the
analytically continued resolvent, and embedded eigenvalues, to
analysis of discrete eigenvalues of a suitable
\emph{non-selfadjoint} operator. As we will see below the
complex-scaling approach suits perfectly, in particular, to our
convergence analysis.

In this section, we assume that the metric graph is finite, but
non-compact, i.e.
\begin{equation}
  \label{eq:ass.res}
  |E_\inl| < \infty
       \und
  0 < |E_\ext| < \infty\,,
\end{equation}
which means, in particular, that the
assumptions~\eqref{eq:deg.len} and~\eqref{eq:vol.ew} are
satisfied.

\subsection{Exterior scaling}
\label{sec:ext}
We decompose the metric graph $X_0$ and the graph-like manifold
$X_\eps$ into an \emph{interior} and \emph{exterior} part
$X_{\eps,\inl}$ and $X_{\eps,\ext}$, respectively. For technical
reasons, it is easier to do the cut not at the initial vertices
$\bd e$ of an external edge $e \in E_\ext$, but at a fixed
distance, say one, from $\bd e$ along $e$, and similarly for the
graph-like manifold.  We therefore consider the \emph{internal}
metric graph $X_{0,\inl}$ consisting of all vertices, all edges of
finite length and the edge parts $(0,1)$ for each external edge.
The exterior metric graph $X_{0,\ext}$ is just the disjoint union
of $|E_\ext|$-many copies of a half-line $[0,\infty)$, and we use
the corresponding parametrization on an external edge. In other
words, we do not regard the \emph{boundary points} $\Gamma_0 = \bd
X_{0,\inl} \cap \bd X_{0,\ext}$ as vertices. Similarly, let
$\Gamma_\eps$ be the common boundary of $X_{\eps,\inl}$ and
$X_{\eps, \ext}$; note that $\Gamma_\eps$ is isometric to
$|E_\ext|$-many copies of $(F,\eps^2 h)$.

Now we introduce the exterior dilation operator. For $\theta \in
\R$ we define by
\begin{align*}
  (U^\theta_0 f)(x) &:=
  \begin{cases}
    f(x),                        & \,\,\; x \in X_{0,\inl}\\
    \e^{\theta/2} f(\e^\theta x),& \,\,\; x \in X_{0,\ext}
  \end{cases} \qquad \text{and}\\
  (U^\theta_\eps u)(z) &:=
  \begin{cases}
    u(z), & z \in X_{\eps,\inl}\\
    \e^{\theta/2} u(\e^\theta x,y),& z=(x,y) \in X_{\eps,\ext}
  \end{cases}
\end{align*}
one-parameter unitary groups on $\HS_0=\Lsqr {X_0}$ and
$\HS_\eps=\Lsqr{X_\eps}$, respectively, acting non-trivially on
the \emph{external} part only.  We call the operator
\begin{equation*}
    H^\theta_\eps
    := U^\theta_\eps \laplacian{X_\eps} U^{-\theta}_\eps
\end{equation*}
for $\eps \ge 0$ the \emph{dilated} Laplacian on $X_\eps$ with the
domain $\dom H^\theta_\eps := U^\theta_\eps(\HS_\eps^2)$ for
\emph{real} $\theta$. A simple calculation shows that
$(H^\theta_\eps u)_e = (\laplacian {X_\eps} u)_e$ holds for
internal edges and that
\begin{equation}
  \label{eq:h.ext}
  (H^\theta_0 f)_e     = -\e^{-2\theta} f_e''
       \und
  (H_\eps^\theta u)_e = -\e^{-2\theta}
      \partial_{xx} u_e + \frac 1 {\eps^2} \Delta_F u_e
\end{equation}
is true for external edges with the domain $\HS^{2,\theta}_\eps:=
\dom H^\theta_\eps$ given by
\begin{equation}
  \label{eq:def.dil.op}
  \HS^{2,\theta}_\eps := \Bigset{f \in
      \Sob[2]{X_{\eps,\inl}} \oplus \Sob[2]{X_{\eps,\ext}}}
    { 
          u_\ext  = \e^{\theta/2} u_\inl, \quad
          \vec u'_\ext  = \e^{3\theta/2} \vec u'_\inl
     \text{ on $\Gamma_\eps$}}.
\end{equation}
Here, $\Sob[2]{X_{\eps,\bullet}}$ are the functions from $\dom
\laplacian {X_\eps}$ restricted to $X_{\eps,\bullet}$, i.e.
without any condition at $\Gamma_\eps$. In addition, $\vec
u'_\bullet$ is defined as the longitudinal derivative on the
common boundary $\Gamma_\eps$ oriented \emph{away} from the
internal part.

The expression of $H^\theta_\eps$ now can be generalized to
\emph{complex} $\theta$ in the strip
\begin{equation*}
  S_\vartheta
   =  \bigset{\theta \in \C}{ |\Im \theta| < \vartheta/2}
\end{equation*}
where $0 \le \vartheta < \pi$; we call the corresponding
$H^\theta_\eps$ the \emph{complex dilated} Laplacian. The
operators $\{H^\theta_\eps\}_\theta$ form a family with spectrum
contained in the common sector
\begin{equation*}
  \Sigma_\vartheta
    := \bigset{ z \in \C}{ |\arg z| \le \vartheta}.
\end{equation*}
Moreover, we can determine the essential spectrum coming from the
external part. Note that the branches of the essential spectrum
associated with the higher transverse eigenvalues $\EW k F,\; k
\ge 2$, on the graph-like manifold all vanish as $\eps\to 0$.
\begin{lemma}
  \label{lem:ess.sp}
  Let the metric graph be finite and non-compact, i.e.
  \eqref{eq:ass.res} is fulfilled. Then
  \begin{equation*}
    \essspec {H^\theta_0} = \e^{-2 \theta} [0, \infty)
       \und
     \essspec {H_\eps^\theta} =
     \frac 1 {\eps^2} \bigcup_{k \in \N}
           \bigl(\EW k F + \e^{-2 \theta} [0,\infty)\bigr).
  \end{equation*}
  In particular, since $\EW 1 F = 0$, we have
  \begin{equation*}
    \essspec {H_\eps^\theta} \cap B
        = \e^{-2\theta}[0,\infty) \cap B
  \end{equation*}
  for any bounded set $B \subset \C$ provided $\eps>0$ is small
  enough.
\end{lemma}

In addition, we have the following important result:
\begin{theorem}
  \label{thm:analytic}
  For $z$ not contained in the $\vartheta$-sector $\Sigma_\vartheta$, the
  resolvents
  \begin{equation*}
    R^\theta_\eps(z):=(H^\theta_\eps-z)^{-1}
  \end{equation*}
  depend analytically on $\theta \in S_\vartheta$.
\end{theorem}
This is a highly non-trivial fact since $H^\theta_\eps$ is neither
of type~A nor of type~B, in other words, both the sesquilinear
form and the operator domains depend on $\theta$ even for real
$\theta$. To put it differently, the (non-smooth) exterior scaling
as defined here is a very singular perturbation of the operator
$\laplacian{X_\eps}=H^0_\eps$. The main idea is to compare
$R^\theta_\eps$ with the resolvent of a \emph{decoupled} operator,
where one imposes Dirichlet boundary conditions at $\Gamma_\eps$.
An inspiration for such an idea was used in~\cite{cdks:87}; for a
full proof in our situation we refer to~\cite{exner-post:07}. As a
consequence, we have the following result on the discrete
spectrum.\footnote{One can derive also other spectral properties,
e.g., $\spec[sc]{H^\theta_\eps}=\emptyset$ -- see
\cite[Prop.~5.8]{exner-post:07} for more details.}
\begin{lemma}
  \label{lem:disc}
  The discrete spectrum of $H^\theta_\eps$ is locally constant in
  $\theta$. As a consequence, discrete
  (complex) eigenvalues are ``revealed'' if $\Im \theta$ is positive
  and large enough. The same is true for eigenvalues embedded into
  the continuous spectrum $[0,\infty)$ of the Laplacian; in this
  case it is sufficient to have $\Im \theta>0$.
\end{lemma}
This is crucial for the above mentioned reformulation. Recall that
by the most common definition a \emph{resonance} is a pole in a
meromorphic continuation of the resolvent over the cut
corresponding to the essential spectrum into the ``unphysical
sheet'' of the Riemann energy surface~\cite{exner:85}. Rotating
the essential spectrum one can reveal these singularities; this
allows us to identify a \emph{resonance} of $H^0_\eps = \laplacian
{X_\eps}$ with a complex $\Lsymb_2$-eigenvalue of the dilated
operator $H^\theta_\eps$ for $\Im \theta$ large enough. Notice
also that such a definition is consistent: it does not depend on
where we cut the spaces into an interior and exterior part; of
course, as far as the interior part remains compact.

\subsection{Convergence result}
\label{sec:conv.reson}
In order to demonstrate convergence properties of resonances, we
need to introduce a scale of Hilbert spaces associated to the
non-selfadjoint operator $H^\theta_\eps$. In particular, we set
$\HS^{2,\theta}:=\dom H^\theta_\eps$ with the norm
$\norm[2,\theta] u := \norm{(H^\theta_\eps +1)u}$ and
$\HS^{-2,\theta}:=(H^{\conj \theta}_\eps)^*$ with the dual norm;
for details we refer to~\cite[App.~A]{exner-post:07}.
\begin{definition}
  We say that an operator $\map J {\HS_0} {\HS_\eps}$ is
  \emph{$\delta$-quasi-unitary} with respect to $H^\theta_\eps$, iff
  $J^*J=\id_0$, $\norm J = 1$, and
  \begin{equation*}
    \norm[2,\theta \to 0] {JJ^*-\id_\eps}
    = \norm{(JJ^*-\id_\eps)R_\eps^\theta)} \le \delta,
  \end{equation*}
  where $R_\eps^\theta:=(H_\eps^\theta+1)^{-1}$ and
  $\norm[2,\theta \to 0] A$ is the operator norm of $\map A
  {\HS_\eps^{2,\theta}} {\HS_\eps^0}$.
\end{definition}
It is much easier to use the $\delta$-quasi-unitarity with respect
to the non-dilated operator (via its quadratic form) as in
\Def{quasi}. In particular, we would like to compare the
non-dilated scale $\HS^1_\eps$ of order $1$ and the dilated scale
$\HS^{2,\theta}_\eps$ of order $2$:
\begin{definition}  $\Delta_\eps$ and $H^\theta_\eps$ are \emph{compatible}
  if there is a family of bounded, invertible operators
  $T_\eps^\theta$ on $\HS_\eps$ such that
  $T_\eps^{-\theta}=(T_\eps^{\theta})^{-1}$, $T_\eps^{\conj
    \theta}=(T_\eps^\theta)^*$ and $T_\eps^{-\theta}(\HS_\eps^1)
  \subset \HS_\eps^{2,\theta}$.
  Given $\eps \ge 0$, we say that $\Delta_\eps$ and $H^\theta_\eps$ are
  \emph{uniformly} compatible with respect to~$\eps$, if there is a
  constant $C^\theta$, \emph{independent} of $\eps$, such that
  \begin{equation*}
    \norm[2,\theta \to 1] {T_\eps^\theta}
    = \norm{(\Delta_\eps+1)^{1/2}
           T_\eps^\theta R_\eps^\theta} \le C^\theta
      \und
    \norm{R_\eps^\theta} \le C^\theta.
  \end{equation*}
\end{definition}
\smallskip

In the situation we consider there is a natural candidate for
$T^\theta_\eps$, namely
\begin{equation*}
  T^\theta_\eps u
  := u_\inl \oplus \e^{\theta/2} u_\ext.
\end{equation*}
In contrast to $U_\eps^\theta$, the operators $T^\theta_\eps$ are
defined also for \emph{complex} values of $\theta$.  Again, as for
the analyticity, the proof of uniform compatibility in our example
needs some technical preliminaries which we skip here. In essence,
one needs to define the resolvent $R_\eps^\theta$ as a bounded
operator from $\HS^{-1,\theta}$ to $\HS^{1,\theta}$ with
$\eps$-independent bound, where
\begin{equation*}
  \HS^{1,\theta}:=T^{\theta}(\HS^1)
  = \Bigset{f \in
      \Sob{X_{\eps,\inl}} \oplus \Sob{X_{\eps,\ext}}}
    { \:u_\ext  = \e^{\theta/2} u_\inl \text{ on $\Gamma_\eps$}}
\end{equation*}
is an appropriate space of order $1$ for $H_\eps^\theta$, not
necessarily related to $\dom (H_\eps^\theta+1)^{1/2}$.
\begin{lemma}
  \label{lem:comp.dil}
  For a given $\eps>0$, the operators $\laplacian {X_\eps}$ and its
  complex dilated counterparts $H_\eps^\theta$ are (uniformly)
  compatible. In particular, the map $J$ as defined in~\eqref{eq:j} is
  $\Err(\eps^{1/2})$-quasi-unitary with respect to $H_\eps^\theta$,
  where $\Err(\eps^{1/2})$ depends on $\theta$.
\end{lemma}
\begin{proof}
  The compatibility is demonstrated
  in~\cite[App.~C]{exner-post:07}. The second assertion follows from
  \begin{align*}
   \norm{(JJ^*-\id_\eps)R_\eps^\theta)}
   &\le \norm{T_\eps^{-\theta}}  \cdot
     \norm{(JJ^*-\id_\eps)(\Delta_\eps+1)^{-1/2})}
     \cdot  \norm{(\Delta_\eps+1)^{1/2}
           T_\eps^\theta R_\eps^\theta} \\
   &\le \e^{-\Re \theta/2} \Err(\eps^{1/2}) C^\theta
  \end{align*}
  where we employed the fact that $T_\eps^\theta J = J T_0^\theta$ and
  \Lem{j.quasi}.
\end{proof}
Now we are in position to state our first convergence result of
this section:
\begin{theorem}
  \label{thm:res.dil}
  Assume that the metric graph $X_0$ is finite and non-compact,
  cf.~\eqref{eq:ass.res}. Then the complex dilated Laplacians
  $H_0^\theta$ and $H_\eps^\theta$ are $\Err(\eps^{1/2})$-close, i.e.
  \begin{equation*}
    \norm[2,\theta \to -2,\theta] {JH_0^\theta - H_\eps^\theta J}
    = \norm{R_\eps^\theta J - J R_0^\theta}
    = \Err(\eps^{1/2}),
  \end{equation*}
  where $\Err(\eps^{1/2})$ depends on $\theta$.
\end{theorem}
\begin{proof}
  The proof is essentially the same as for the non-dilated case. First we
  define operators on the scales of order one, specifically
  \begin{equation*}
    \map {J^{1,\theta}}{\HS_0^{1,\theta}}{\HS_\eps^{1,\theta}}
        \und
    \map {J^{1,\theta}{}'}{\HS_\eps^{1,\theta}}{\HS_0^{1,\theta}}
  \end{equation*}
  in exactly the same way as in~\eqref{eq:j1}
  and~\eqref{eq:j1'}. Recall that we do not consider the boundary
  points between the internal and external parts as vertices, i.e.
  the graph-like manifold does not have a vertex neighborhood
  there. Hence
  \begin{equation*}
    T_\eps^\theta J^{1,\theta}= J^1 T_0^\theta
      \und
    T_0^\theta J^{1,\theta}{}'= J^1{}' T_\eps^\theta,
  \end{equation*}
  and then we have
  \begin{equation*}
       R_\eps^\theta J - J R_0^\theta
     = R_\eps^\theta
        \bigl[  (J-J^{-1,\theta}) H_0^\theta
              + (J^{-1,\theta} H_0^\theta - H_\eps^\theta J^{1,\theta})
              + H_\eps^\theta (J^{1,\theta} - J) \bigr] R_0^\theta
  \end{equation*}
  where $J^{-1,\theta}:= (J^{1,\conj \theta}{}')^*$.
  The last difference at the right-hand side can be estimated by
  \begin{multline*}
    \norm{R_\eps^\theta H_\eps^\theta (J^{1,\theta} - J) R_0^\theta}
    \le \norm{\id_\eps - R_\eps^\theta} \cdot
        \norm{(J^{1,\theta} - J) R_0^\theta}\\
    \le (1 + \norm{R_\eps^\theta})
        \norm{T_\eps^{-\theta}}  \cdot
        \norm{(J^1 - J) (\Delta_0+1)^{-1/2}}
        \norm{(\Delta_0+1)^{1/2} T_0^\theta R_0^\theta}\\
    \le (1+C^\theta) \e^{-\Re \theta/2}\Err(\eps^{1/2}) C^\theta
  \end{multline*}
  using \Lem{j.comp} and \Lem{comp.dil}, and the first one can be
  treated similarly. For the remaining term we observe that in order
  to prove
  \begin{equation*}
    \norm{R_\eps^\theta (J^{-1,\theta} H_0^\theta -
    H_\eps^\theta J^{1,\theta}) R_0^\theta}
    =\norm[2,\theta \to -2,\theta]
            {J^{-1,\theta} H_0^\theta - H_\eps^\theta J^{1,\theta}}
    \le \delta
  \end{equation*}
  it suffices to show
  \begin{equation}
    \label{eq:j.comm.dil}
    \bigl|  \iprod {H^{\conj \theta}_0 f}{J^{1,\theta}{}' u}
          - \iprod {J^{1,\theta} f} {H^\theta_\eps u}  \bigr|
    \le \wt \delta \norm[1]{T^{\conj\theta}_0 f}
                   \norm[1]{T^\theta_\eps u}
  \end{equation}
  for $f \in \HS^{2,\conj \theta}_0$ and $u \in \HS^{2,\theta}_\eps$.
  From the compatibility, we obtain $\norm[1]{T^\theta_\eps u} \le
  C^\theta \norm[2,\theta] u$ and similarly for $f$; in particular, we
  can choose $\delta=(C^\theta)^2\wt \delta$, however, the
  estimate~\eqref{eq:j.comm.dil} is almost the same as in the
  non-dilated case given in \Lem{j.comm}.
\end{proof}
As in the non-dilated case, on can develop a functional calculus
for the pairs of operators $H^\theta_0$ and $H^\theta_\eps$ --
cf.~\cite[App.~B]{exner-post:07}. Since now the operators are not
self-ajoint, we only have a \emph{holomorphic} functional
calculus. In particular, we can show
\begin{equation*}
  \norm{ \1_D(H_\eps^\theta) J
    - J \1_D(H_0^\theta)}
  \le \Err(\eps^{1/2})
  \quad \text{and} \quad
  \norm{ \1_D(H_\eps^\theta)
    - J \1_D(H_0^\theta)J^*}
  \le \Err(\eps^{1/2})
\end{equation*}
for the spectral projections, provided $D$ is an open disc
containing a single discrete eigenvalue $\lambda(0)$ of
$H_0^\theta$. From here our main result on resonances follows:
\begin{theorem}
  \label{thm:resonances}
  Assume that the metric graph $X_0$ is finite and non-compact,
  cf.~\eqref{eq:ass.res}.  If $\lambda(0)$ is a resonance of the
  Laplacian $\laplacian{X_0}$ with a multiplicity $m>0$, then for a
  sufficiently small $\eps>0$ there exist~$m$ resonances
  $\lambda_1(\eps), \dots, \lambda_m(\eps)$ of $\laplacian{X_\eps}$,
  satisfying $\Im\lambda_j(\eps) < 0$ and not necessarily mutually different,
  which all converge to $\lambda(0)$ as $\eps\to 0$. The same is true
  in the case when $\lambda(0)$ is an embedded eigenvalue of
  $\laplacian{X_0}$, except that then only $\Im\lambda_j(\eps)\le 0$ holds
  in general.

  Finally, if the multiplicity of $\lambda(0)$ is one with a normalized
  eigenfunction $u^\theta(0)$ (corresponding to a \emph{resonance}
  or \emph{embedded eigenvalue}
  for $\laplacian{X_0}$), then there exists a normalized eigenfunction
  $u^\theta(\eps)$ (related to the respective entity for
  $\laplacian{X_\eps}$) on the graph-like manifold) such that
  \begin{equation*}
    \norm{J u^\theta(0) - u^\theta(\eps)} = \Err(\eps^{1/2}).
  \end{equation*}
\end{theorem}


\section*{Acknowledgments}
The first author acknowledges a partial support by GAAS and MEYS
of the Czech Republic under projects A100480501 and LC06002, the
second one by DFG under the grant Po-1034/1-1.



\providecommand{\bysame}{\leavevmode\hbox to3em{\hrulefill}\thinspace}

\end{document}